\documentclass[pra,preprint,showpacs,preprintnumbers,superscriptaddress]{revtex4}
\usepackage{bm}
\usepackage{amsthm,amsfonts,amsmath}

\def\eg{{\em e.g.}}
\def\ie{{\em i.e.}}
\def\etal{{\em et al.}}

\newcommand{\ket}[1]{|#1\rangle}
\newcommand{\bra}[1]{\langle#1|}
\newcommand{\braket}[2]{\langle#1|#2\rangle}
\newcommand{\proj}[1]{|#1\rangle\langle#1|}
\newcommand{\set}[1]{\left\{#1\right\}}
\newcommand{\abs}[1]{\left|#1\right|}

\newcommand{\ct}{^{\dagger}}
\newcommand{\tp}{^{\mathsf{T}}}

\newcommand{\ofs}{}

\newcommand{\krange}{{k:\,\lambda_k(s) \neq 1}}

\newcommand{\note}[1]{}

\newcommand{\Hi}{\mathcal{H}}
\newcommand{\Hr}{\mathcal{R}}

\newcommand{\mx}[1]{\begin{pmatrix}#1\end{pmatrix}}
\newcommand{\smx}[1]{\bigl(\begin{smallmatrix}#1\end{smallmatrix}\bigr)}
\newcommand{\rv}[2]{(#1 \;\, #2)}
\newcommand{\co}[2]{[#1,#2)}

\DeclareMathOperator{\spn}{span}
\DeclareMathOperator{\diag}{diag}

\DeclareMathOperator{\HT}{HT}
\DeclareMathOperator{\MT}{MT}

\newcommand{\AdiabaticAlgorithm}{\mathbf{Adiabatic\ Search\ Algorithm}}
\newcommand{\RandomWalkAlgorithm}{\mathbf{Random\ Walk\ Algorithm}}

\newtheorem{theorem}{Theorem}

\newtheorem{lemma}{Lemma}

\theoremstyle{plain}
\newtheorem*{theorem*}{Theorem}

\newsavebox{\fmbox}
\newenvironment{fmpage}[1]
     {\medskip\begin{lrbox}{\fmbox}\begin{minipage}{#1}}
     {\end{minipage}\end{lrbox}\fbox{\usebox{\fmbox}}\medskip}
\newcommand{\algorithm}[1]{
\begin{center}
\begin{fmpage}{11.5cm}
#1
\end{fmpage}
\end{center}}

\begin{document}





\title{On the adiabatic condition and\\the quantum hitting time of Markov chains}


\author{Hari Krovi}\email{krovi@engr.uconn.edu}
\altaffiliation[Present affiliation: ]{University of Connecticut}
\affiliation{NEC Laboratories America, Inc.}
\author{Maris Ozols}\email{marozols@yahoo.com}
\affiliation{NEC Laboratories America, Inc.}
\affiliation{University of Waterloo and Institute for Quantum Computing}
\author{J\'er\'emie Roland}\email{jroland@nec-labs.com}
\affiliation{NEC Laboratories America, Inc.}










\begin{abstract}
\noindent We present an adiabatic quantum algorithm for the abstract problem of searching marked vertices in a graph, or spatial search.
Given a random walk (or Markov chain) $P$ on a graph with a set of unknown marked vertices, one can define a related absorbing walk $P'$ where outgoing transitions from marked vertices are replaced by self-loops.
We build a Hamiltonian $H(s)$ from the interpolated Markov chain $P(s)=(1-s)P+sP'$ and use it in an adiabatic quantum algorithm to drive an initial superposition over all vertices to a superposition over marked vertices.
The adiabatic condition implies that for any reversible Markov chain and any set of marked vertices, the running time of the adiabatic algorithm is given by the square root of the classical hitting time. This algorithm therefore demonstrates a novel connection between the adiabatic condition and the classical notion of hitting time of a random walk. 
It also significantly extends the scope of previous quantum algorithms for this problem, which could only obtain a full quadratic speed-up for state-transitive reversible Markov chains with a unique marked vertex.
\end{abstract}


\maketitle

\section*{Introduction}

Adiabatic quantum computation was introduced by Farhi \etal\ in \cite{Farhi}. It proceeds as follows. Suppose that the solution of a computational problem can be encoded in the ground state of a problem Hamiltonian $H_P$. We start in the ground state of an initial Hamiltonian $H_0$ which is easy to construct. Then we slowly change the Hamiltonian from $H_0$ to $H_P$, so that the instantaneous Hamiltonian at any point in the evolution is $H(s)=(1-s)H_0+sH_P$, where $0\leq s \leq 1$. If this is done slowly enough, then the Adiabatic Theorem of Quantum Mechanics \cite{Mes} guarantees that the state at all points in the evolution stays close to the ground state of $H(s)$. Note that the validity of the \emph{folk} version of the Adiabatic Theorem, as stated in many books of Quantum Mechanics such as \cite{Mes}, has recently been the subject of much debate~\cite{marzlin-sanders,sarandy,wu,tong}, and there is no rigorous proof of it that holds under full generality. It is nevertheless possible to state a more rigorous version of the theorem, so that the \emph{folk} adiabatic condition can be proved to be sufficient in many interesting cases~\cite{jrs07}, such as the adiabatic version of Grover's algorithm~\cite{grover96,vdam01,rc02}.


Many classical randomized algorithms rely heavily on random walks or Markov chains. The notion of hitting time is a useful characterization of Markov chains used when searching for a marked vertex. Quantum walks are natural generalizations of classical random walks.
The notion of hitting time has been carried over to the quantum case in \cite{AKR,Kempe,Sze,KB,MNRS,MNRS2,VKB}, by generalizing the classical notion in different ways. It is intimately related to the problem of spatial search. Suppose that we are given a graph where some vertices are marked. Classically, a simple algorithm to find a marked vertex is to repeatedly apply some random walk $P$ on the graph until one of the marked vertices is reached. The hitting time of $P$ is precisely the expected number of repetitions necessary to reach a marked vertex, starting from the stationary distribution of $P$. The notions of quantum hitting time are based on different quantum analogues of this algorithm. They usually show some quadratic improvement of the quantum hitting time over the classical hitting time. However, until the present paper, they could only show such a relation under restricted conditions: either the quantum algorithm could only detect marked elements~\cite{Sze}, or it could only be applied to state-transitive reversible Markov chains with a unique marked element~\cite{MNRS2}. Whether this quadratic speed-up for finding a marked element also holds for any reversible Markov chain and for multiple marked elements was an open question. In this paper, we answer this question by the positive, by providing an adiabatic quantum algorithm for this problem. In addition to being more general, the algorithm is also conceptually very clean, it implements a simple rotation in a two dimensional subspace based on a quantum walk on the edges of the graph. Moreover, it reveals a close connection between the adiabatic condition and the notion of hitting time.


The paper is structured as follows. In Section~\ref{sect:related-work} we describe related work and in Section~\ref{sect:main-result} we state our main result. In Section~\ref{sect:Preliminaries} we introduce the necessary concepts such as Markov chains, the discriminant matrix, and the quantization of Markov chains. In Section~\ref{sect:Spectrum} we evaluate the spectrum of the interpolating Hamiltonian and in Section~\ref{sect:Adiabatic theorem} we impose the adiabatic condition to calculate the running time of the adiabatic quantum algorithm. In Section~\ref{sect:Running time} we relate this to the hitting time of the Markov chain we started from, and show that the running time of the adiabatic evolution is at most the square root of the classical hitting time.

\section{Related work} \label{sect:related-work}
Inspired by Ambainis' quantum walk algorithm for Element Distinctness~\cite{Ambainis04}, Szegedy~\cite{Sze} introduced a powerful way of quantizing Markov chains which led to new quantum walk-based algorithms. He showed that for any symmetric Markov chain a quantum walk could detect the presence of marked vertices in at most the square root of the classical hitting time. However, showing that a marked vertex could also be found in the same time (as is the case for the classical algorithm) proved to be a very difficult task. Magniez~\etal~\cite{MNRS} extended Szegedy's approach to the larger class of ergodic Markov chains, and proposed a quantum walk-based algorithm to find a marked vertex, but its complexity may be larger than the square root of the classical hitting time. A typical example where their approach fails to provide a quadratic speed-up is the 2D grid, where their algorithm has complexity $\Theta(n)$, whereas the classical hitting time is $\Theta(n\log n)$. Ambainis~\etal~\cite{AKR} and Szegedy's~\cite{Sze} approaches yield a complexity of $\Theta(\sqrt{n}\cdot\log n)$ in this special case, for a unique marked vertex. Childs and Goldstone~\cite{childs1,childs2} also obtained a similar result using a continuous-time quantum walk.
However, whether a full quadratic speed-up was possible remained an open question, until Tulsi~\cite{Tulsi} proposed a solution involving a new technique. Magniez~\etal~\cite{MNRS2} extended Tulsi's technique to any reversible state-transitive Markov chain, showing that for such chains, it is possible to find a unique marked vertex with a full quadratic speed-up over the classical hitting time. However, the state-transitivity is a strong symmetry condition, and furthermore their technique cannot deal with multiple marked vertices. It would be strange if one had to rely on involved techniques to solve the finding problem under such restricted conditions, while the classical analogue algorithm is conceptually very simple and works under very general conditions. 

In this paper we show that these issues can be resolved by combining the idea of the quantization of Markov chains with adiabatic quantum computation (note that Childs and Goldstone~\cite{childs1,childs2} showed that their algorithm for spatial search on the grid could also be translated into an adiabatic algorithm, but this failed to give a quadratic speed-up for low dimensions).

\section{Main result} \label{sect:main-result}

Before describing our quantum algorithm, let us first recall the classical algorithm on which it will provide a quadratic speed-up. This very simple algorithm just consists in walking randomly on the graph until a marked vertex is reached. More precisely, it relies on a Markov chain $P$ with stationary distribution $\pi$, and works as follows.

\algorithm{
$\RandomWalkAlgorithm$
\begin{enumerate}
 \item Sample a vertex $x \in X$ according to distribution $\pi$.
 \item Check if $x$ is marked.\label{item:check}
 \item If so, output $x$.
 \item Otherwise, update $x$ according to $P$, and go back to step~\ref{item:check}.\label{item:walk}
\end{enumerate}
}

Let $P$ be an ergodic Markov chain, and $M$ be a set of marked vertices. The \emph{hitting time} of $P$ with respect to $M$, denoted by $\HT(P,M)$, is the expected number of executions of step~\ref{item:walk} during the course of the $\RandomWalkAlgorithm$ (where the expectation is calculated conditionally on the initial vertex being unmarked). Note that since the algorithm stops as soon as a marked element is reached, this is equivalent to using an \emph{absorbing} Markov chain $P'$, which acts as $P$ on all but marked vertices, where all outgoing transitions are replaced by self-loops.

Previous attempts at providing a quantum speed-up over this classical algorithm have followed one of these two approaches:
\begin{itemize}
 \item Combining a quantum version of $P$ with a reflection through marked vertices to mimic a Grover operation~\cite{AKR,Ambainis04,MNRS}.
 \item Directly applying a quantum version of $P'$~\cite{Sze,MNRS2}.
\end{itemize}
The problem with these approaches is that they would only be able to find marked vertices in very restricted cases. We explain this by the different nature of random and quantum walks: while both accept a stable state, \ie, the stationary distribution for the random walk and the eigenstate with eigenvalue 1 for the quantum walk, the way both walks act on other states is dramatically different. Indeed, an ergodic random walk will converge to its stationary distribution from any initial distribution. This apparent robustness may be attributed to the inherent randomness of the walk, which will smooth out any initial perturbation. After many iterations of the walk, non-stationary contributions of the initial distribution will be damped and only the stationary distribution will survive (this can be attributed to the thermodynamical irreversibility~\footnote{Note that when we consider \emph{reversible Markov chains} as defined in Section~\ref{sect:Reversibility}, this corresponds to a different notion of reversibility than in the usual thermodynamical sense. Actually, even a ``reversible" Markov chain is  thermodynamically irreversible.} of ergodic random walks). On the other hand, this is not true for quantum walks, because in the absence of measurements a unitary evolution is deterministic (and in particular thermodynamically reversible): the contributions of the other eigenstates will not be damped but just oscillate with different frequencies, so that the overall evolution is quasi-periodic. As a consequence, while iterations of $P'$ always lead to a marked vertex, it may happen that iterations of the quantization of $P'$ will never lead to a state with a large overlap over marked vertices, unless the walk exhibits a strong symmetry (as is the case for a state-transitive walk with only one marked element, which could be addressed by previous approaches).

The main new idea of our approach is that, instead of using a quantization of $P$ or $P'$, we first modify the classical random walk, and then use a quantization of the modified walk. The original classical algorithm consists in applying $P'$ on the stationary distribution $\pi$ of $P$. While doing so, the system is brought far from equilibrium since $\pi$ is far (in statistical distance) from any stationary distribution of $P'$, which only have support on marked elements. The random walk will damp any non-stationary contribution of the initial distribution, but a quantum walk based on $P$ or $P'$, which is deterministic until a measurement, seems to have trouble with it. There is however a situation in Quantum Mechanics where contributions from other eigenstates will also cancel out, similarly to what happens for a random walk: if the system starts in a state close to the ground state of its instantaneous Hamiltonian $H(s)$ (\ie, close to equilibrium), and this Hamiltonian varies slowly, the Adiabatic Theorem ensures that the contributions from excited states will cancel out so that the system will remain close to its ground state. Therefore, our strategy is to first modify the classical algorithm so that the system stays close to equilibrium throughout the evolution, and then to translate it into an adiabatic quantum algorithm.

Consider the interpolated Markov chain $P(s)=(1-s)P+sP'$ (see Section~\ref{sect:Interpolation}). Our goal is to drive the stationary distribution $\pi$ of $P$ towards a stationary distribution of $P'$. Instead of immediately applying $P'$ on $\pi$, we could rather apply $P(s)$ while slowly switching $s$ from $0$ to $1$, so that the system remains at all time close to the stationary distribution $\pi(s)$ of $P(s)$. It can be shown that this leads to an algorithm with only a constant overhead with respect to $\RandomWalkAlgorithm$. Therefore, this new classical algorithm still runs in time $O(\HT(P,M))$, but the difference is that at all time the system is close to equilibrium, so that we are in a better shape for designing a quantum analogue based on the Adiabatic Theorem.

Using a Hamiltonian version of Szegedy's quantization technique, proposed by Somma and Ortiz~\cite{Somma}, we map $P(s)$ to a Hamiltonian $H(s)$ on the edge space $\Hi \otimes \Hi$, where $\Hi$ is the Hilbert space whose basis states are labeled by the vertices of the graph (see Section~\ref{sect:Hamiltonian}). The eigenstate of $H(s)$ with eigenvalue 0 then corresponds to the stationary distribution of $P$ for $s=0$, and to a distribution over marked vertices for $s=1$, so that this Hamiltonian may be used to solve the search problem by adiabatic evolution. The algorithm consists in preparing the first register in the state $\ket{\pi}$ corresponding to the stationary distribution $\pi$ of $P$ and applying the Hamiltonian $H(s)$ using a schedule $s(t)$ (we will specify $s(t)$ explicitly later). The resulting adiabatic evolution effectively implements a rotation on the first register at constant angular velocity from $\ket{\pi}$ to a superposition over marked vertices.
\algorithm{
$\AdiabaticAlgorithm$
\begin{enumerate}
 \item Prepare the state $\ket{\pi}\ket{0}$.
 \item Apply the time-dependent Hamiltonian $H(s)$ with schedule $s(t)$ from $t=0$ to $t=T$, where $T =  \frac{\pi}{2\varepsilon} \sqrt{\HT(P,M)}$.
 \item Measure the first register in the vertex basis.
\end{enumerate}
}
Under the assumption that the \emph{folk} adiabatic condition holds in our setup, we prove the following:
\begin{theorem}\label{thm:main}
For any ergodic and reversible Markov chain $P$ with a set of marked vertices $M$, the $\AdiabaticAlgorithm$ finds a marked vertex with probability at least $1-\varepsilon^2$ in a time $T = \frac{\pi}{2\varepsilon} \sqrt{\HT(P,M)}$, where $\HT(P,M)$ is the hitting time of the classical Markov chain $P$ with respect to the set of marked vertices $M$.
\end{theorem}

This theorem constitutes our main result and the body of this paper will be devoted to its proof.
While it relies on the \emph{folk} adiabatic condition, a similar statement can be made for a related quantum circuit algorithm, where no such condition is necessary, as shown in~\cite{kmor10}. Nevertheless, as explained in~\cite{kmor10}, the intuition behind the quantum circuit algorithm originates from the present adiabatic quantum algorithm.

\section{Preliminaries} \label{sect:Preliminaries}

\subsection{Classical interpolation} \label{sect:Interpolation}

Let us consider a Markov chain on a discrete state space $X$ of size $n$, and let $P$ be the $n \times n$ (row) stochastic matrix~\footnote{Throughout the paper we use the convention that each row of a stochastic matrix $P$ sums to one ($\forall x \in X: \sum_{y \in X} P_{xy} = 1$) and probability distributions are row vectors and hence are multiplied to the transition matrix from the left hand side (\eg, $\pi P$).} describing the transition probabilities of the Markov chain.
From now on, we will assume that the Markov chain is \emph{ergodic}, meaning that it is both \emph{irreducible} (any state can be reached from any other state by a finite number of steps) and \emph{aperiodic} (there is no integer $k > 1$ that divides the length of every cycle of the underlying directed graph of the stochastic matrix $P$). Assume that a subset $M\subset X$ of elements are marked and let $m$ be the size of $M$. Let $P'$ be the Markov chain obtained from $P$ by turning all outgoing transitions from marked elements into self-loops. We call $P'$ the \emph{absorbing} version of $P$. Note that $P'$ differs from $P$ only in the rows corresponding to the marked elements (where it contains all zeros on non-diagonal elements, and ones on the diagonal). If we arrange the elements of $X$ so that the marked elements are the last ones, matrices $P$ and $P'$ have the following block structure:
\begin{align}
  P &:= \mx{
    P_{UU} & P_{UM} \\
    P_{MU} & P_{MM}
  }, &
  P' &:= \mx{
    P_{UU} & P_{UM} \\
    0      & I
  },
  \label{eq:P}
\end{align}
where $P_{UU}$ and $P_{MM}$ are square matrices of size $(n-m) \times (n-m)$ and $m\times m$, respectively, while $P_{UM}$ and $P_{MU}$ are matrices of size $(n-m) \times m$ and $m \times (n-m)$, respectively. We call
\begin{equation}
  P(s) := (1-s) P + s P', \quad 0 \leq s \leq 1,
  \label{eq:P(s)}
\end{equation}
the \emph{classical interpolation} of $P$ and $P'$. Note that $P(0) = P$, $P(1) = P'$, and $P(s)$ has block structure
\begin{equation}
  P(s) = \mx{
    P_{UU}      & P_{UM} \\
    (1-s)P_{MU} & (1-s)P_{MM} + s I
  }.
  \label{eq:Block P(s)}
\end{equation}
Moreover, note that the ergodicity of $P$ implies that $P(s)$ is also ergodic, except for $s = 1$.

\subsection{Stationary distribution and reversibility} \label{sect:Reversibility}

By definition, since $P$ is ergodic, it has a unique and non-vanishing \emph{stationary distribution}, \ie, a probability distribution $\pi$ such that $\pi P = \pi$. An ergodic Markov chain $P$ is called \emph{reversible} if it satisfies the  so-called \emph{detailed balance condition}
\begin{equation}
  \forall x,y \in X: \pi_x P_{xy} = \pi_y P_{yx}.
  \label{eq:Detailed balance}
\end{equation}
This implies that for reversible Markov chains, the net flow of probability in the stationary distribution between every pair of states is zero.

From now on we will assume that $P$ is reversible. Let us argue that $P(s)$ is also reversible. Let $\pi := \rv{\pi_U}{\pi_M}$ be the stationary distribution of $P$, where $\pi_U$ and $\pi_M$ are row vectors of length $n-m$ and $m$, respectively. Let $p_M=\sum_{x\in M}\pi_x$ be the probability to pick a marked element from the stationary distribution and  $\pi(s)$ be the following distribution:
\begin{equation}
  \pi(s) := \frac{1}{1 - s (1-p_M)} \mx{(1-s) \pi_U & \pi_M}.
  \label{eq:pi(s)}
\end{equation}
One can check that $\pi(s)$ is a stationary distribution for $P(s)$ for any $s\in [0,1]$. Moreover, for $s \in \co{0}{1}$ $P(s)$ is ergodic and this is therefore the unique stationary distribution, while for $s=1$ any distribution which only has support on marked vertices is stationary.
Equation (\ref{eq:Detailed balance}) can be used to show that
\begin{equation}
  \forall s \in [0,1], \forall x,y \in X: \pi_x(s) P_{xy}(s) = \pi_y(s) P_{yx}(s),
  \label{eq:Extended detailed balance}
\end{equation}
which means that $P(s)$ is reversible for $s \in \co{0}{1}$. Condition (\ref{eq:Extended detailed balance}) is also satisfied at $s=1$, but $P(1)=P'$ is not ergodic, therefore {\em stricto sensu} $P(1)$ is not reversible.

\subsection{Discriminant matrix} \label{sect:Discriminant}

The \emph{discriminant matrix} of a Markov chain $P(s)$ is
\begin{equation}
  D(s) := \sqrt{P(s) \circ P(s)\tp},
  \label{eq:D(s)}
\end{equation}
where the Hadamard product ``$\circ$'' and the square root is computed entry-wise. We prefer to work with $D(s)$ rather than $P(s)$ since a Markov chain is not necessarily symmetric while its discriminant matrix is.

For a reversible Markov chain, the extended detailed balance condition (\ref{eq:Extended detailed balance}) implies that $D_{xy}(s) = \sqrt{\smash[b]{P_{xy}(s) P_{yx}(s)}} = P_{xy}(s) \sqrt{\pi_x(s) / \pi_y(s)}$ or equivalently
\begin{equation}
  D(s) = \diag \bigl( \! \sqrt{\pi(s)} \, \bigr) \: P(s) \: \diag \bigl( \! \sqrt{\pi(s)} \, \bigr)^{-1}.
  \label{eq:D(s) and P(s)}
\end{equation}
For $s \in \co{0}{1}$ the right-hand side is well-defined so that $D(s)$ and $P(s)$ are similar and therefore have the same eigenvalues. Moreover, the entry-wise square root of the stationary distribution $\sqrt{\pi(s)\tp}$ is the eigenvector of $D(s)$ with eigenvalue $1$.

At $s=1$ the right-hand side of equation (\ref{eq:D(s) and P(s)}) is not well-defined, but it can be shown that both claims still hold by expanding $P(s)$ according to equation (\ref{eq:Block P(s)}) and considering the limit $s \rightarrow 1$. Then equation (\ref{eq:D(s) and P(s)}) becomes
\begin{equation}
  D(1) = \mx{
    \diag \bigl( \! \sqrt{\pi_U} \, \bigr) \: P_{UU} \: \diag \bigl( \! \sqrt{\pi_U} \, \bigr)^{-1} & 0 \\
    0 & I
  }.
  \label{eq:D(1)}
\end{equation}
This implies that $D(1)$ is similar to $\tilde{P} := \smx{P_{UU}&0\\0&I}$, and in turn to $P(1) = \smx{P_{UU}&P_{UM}\\0&I}$ as well. To see that $\sqrt{\pi(1)\tp}$ is an eigenvector of $D(1)$ with eigenvalue $1$, note that $\pi(1) = \rv{0}{\pi_M} / p_M$, and $D(1)$ acts as the identity on marked elements (this follows from equations (\ref{eq:pi(s)}) and (\ref{eq:D(1)}), respectively).



\subsection{The quantum Hamiltonian} \label{sect:Hamiltonian}

Szegedy~\cite{Sze} proposed a general method to map a random walk to a unitary operator that defines a quantum walk. Recently Somma and Ortiz~\cite{Somma} showed how Szegedy's method may be adapted to build a Hamiltonian. We apply this method to the random walk $P(s)$.

The first step is to map the rows of $P(s)$ to quantum states. Let $\Hi$ be a Hilbert space of dimension $n=|X|$. For every $x \in X$ we define the following state in $\Hi$:
 \begin{equation}
  \ket{p_x(s)} := \sum_{y \in X} \sqrt{P_{xy}(s)} \ket{y}.
\end{equation}
Following Szegedy~\cite{Sze}, we then define a unitary operator $V(s)$ acting on $\Hi\otimes\Hi$ as
\begin{equation}
  V(s) \ket{x, 0} := \ket{x} \ket{p_x(s)} = \sum_{y \in X} \sqrt{P_{xy}(s)} \ket{x,y},
  \label{eq:V(s)}
\end{equation}
when the second register is in some reference state $\ket{0}\in\Hi$, and arbitrarily otherwise.

Let $S$ be the gate that swaps both registers.
When restricted to $\ket{0}$ in the second register, the operator
$V\ct(s) S V(s)$ acts as $D(s)$:
\begin{equation}
  \bra{x, 0} V\ct(s) SV(s) \ket{y, 0}
  = \braket{x}{p_y(s)} \braket{p_x(s)}{y}
  = \sqrt{P_{xy}(s) P_{yx}(s)}
  = D_{xy}(s).
  \label{eq:VW-discriminant}
\end{equation}

Following~\cite{Somma}, we now define the Hamiltonian $H(s)$ on $\Hi\otimes\Hi$ as
\begin{equation}
  H(s) := i \bigl[ V\ct(s) SV(s), \Pi_0 \bigr],
\end{equation}
where $\Pi_0 := I \otimes \ket{0} \bra{0}$ is the projector that keeps only the component containing the reference state $\ket{0}$ in the second register and $[A,B]:=AB-BA$ is the commutator.

\section{Spectral decomposition of $H(s)$} \label{sect:Spectrum}

To understand the properties of the Hamiltonian $H(s)$, let us find its spectral decomposition. We will relate its spectrum to that of $D(s)$.

\subsection{Diagonalization of $D(s)$} \label{sect:D(s)}

Recall from equation (\ref{eq:D(s)}) that $D(s)$ is real and symmetric. Therefore, its eigenvalues are real and its eigenvectors form an orthonormal basis of $\Hi$ with real amplitudes. Let
\begin{equation}
  D(s) = \sum_{i=1}^n \lambda_i(s) \ket{v_i(s)} \bra{v_i(s)}
  \label{eq:spectral}
\end{equation}
be the spectral decomposition of $D(s)$. We can make the eigenvalues of $P(s)$ and hence also of $D(s)$ to be non-negative by replacing $P$ with $(P+I)/2$. Note that this will only modify the hitting time of the Markov chain by a factor of $2$. Hence, from now on without loss of generality we assume that all eigenvalues of $D(s)$ are non-negative. In addition, we can arrange them so that
\begin{equation}
  0 \leq \lambda_1(s) \leq \lambda_2(s) \leq \dots \leq \lambda_n(s).
\end{equation}

From the Perron--Frobenius theorem we have that $\forall i: \lambda_i(s) \leq 1$ and $\lambda_n(s) = 1$. In addition, for any $s \in \co{0}{1}$ the Markov chain $P(s)$ is ergodic and $\forall i \neq n: \lambda_i(s) < 1$. On the other hand, at $s = 1$ the Markov chain is not ergodic and has eigenvalue $1$ with multiplicity $m$. We may summarize this as follows:
\begin{align}
  \lambda_{n-1}(s) < \lambda_n(s) = 1&, \quad \forall s \in \co{0}{1}, \\
  \lambda_{n-m}(1) < \lambda_{n-m+1}(1) = \dots = \lambda_n(1) = 1&.
\end{align}

Recall from Section~\ref{sect:Discriminant} that $\sqrt{\pi(s)\tp}$ is an eigenvector of $D(s)$ with eigenvalue $1$ for any $s \in [0,1]$, so let us choose $\ket{v_n(s)}$ in equation (\ref{eq:spectral}) as
\begin{equation}\label{eq:v_n-pi}
  \ket{v_n(s)} := \sqrt{\pi(s)\tp}.
\end{equation}

\subsection{Diagonalization of $H(s)$}


Now, let us express the eigenvalues and eigenvectors of $H(s)$ in terms of those of $D(s)$. First, let us break up the Hilbert space $\Hi\otimes\Hi$ into mutually orthogonal subspaces that are invariant under $H(s)$. Let
 \begin{itemize}
  \item $\mathcal{B}_k(s) := \spn\set{\ket{v_k(s), 0}, V\ct(s) SV(s) \ket{v_k(s), 0}}$, $\forall k \neq n$,
  \item $\mathcal{B}_n(s) := \spn\set{\ket{v_n(s), 0}}$,
  \item $\mathcal{B}^\perp(s) := (\oplus_{k=1}^n \mathcal{B}_k(s))^\perp$.
 \end{itemize}
Note that $\bra{v_k(s),0} V\ct(s) SV(s) \ket{v_k(s),0} = \bra{v_k(s)} D(s) \ket{v_k(s)} = \lambda_k(s)$ by equations (\ref{eq:VW-discriminant}) and (\ref{eq:spectral}).
Recall that for $s<1$, we have $\lambda_k(s) \neq 1$ for any $k \neq n$ and thus
\begin{equation}
  V\ct(s) SV(s) \ket{v_k(s), 0}
    = \lambda_k(s) \ket{v_k(s), 0} + \sqrt{1-\lambda_k(s)^2} \ket{v_k(s), 0}^\perp
\end{equation}
for some~\footnote{Note that $\ket{v_k(s), 0}^\perp$ depends on how the operator $V(s)$ defined in equation in (\ref{eq:V(s)}) is extended to the whole Hilbert space.} unit vector $\ket{v_k(s), 0}^\perp$ orthogonal to $\ket{v_k(s), 0}$ and lying in the subspace $\mathcal{B}_k(s)$. We also define by continuity $\ket{v_k(1), 0}^\perp := \lim_{s\to 1}\ket{v_k(s), 0}^\perp$.


Following Somma and Ortiz, who were in turn relying on Szegedy's work, we may now characterize the spectrum of $H(s)$.

\begin{lemma}[\cite{Sze,Somma}]
$H(s)$ accepts the following eigenvalues and eigenstates.
\begin{itemize}
 \item On $\mathcal{B}_k(s)$, $\forall k \neq n$:
\begin{equation}
  E^\pm_k(s) := \pm \sqrt{1-\lambda_k(s)^2}, \quad
  \ket{\Psi^\pm_k(s)} := \frac{\ket{v_k(s), 0} \pm \ket{v_k(s), 0}^\perp}{\sqrt{2}},
  \label{eq:Ek(s)}
\end{equation}
\item On $\mathcal{B}_n(s)$:
\begin{equation}
  E_n(s) := 0, \quad
  \ket{\Psi_n(s)} := \ket{v_n(s), 0}.
  \label{eq:En(s)}
\end{equation}
\item On $\mathcal{B}^\perp(s)$:
\begin{equation}
  F_j(s) := 0, \quad
  \ket{\Phi_j(s)}, \quad
  j \in \set{1, \dotsc, (n-1)^2},
  \label{eq:Fj(s)}
\end{equation}
where $\set{\ket{\Phi_j(s)}}$ defines an arbitrary basis of $\mathcal{B}^\perp(s)$.
\end{itemize}
\end{lemma}

\begin{proof}
We consider the case $s\in\co{0}{1}$, the case $s=1$ follows by continuity.
Since  $V\ct(s) SV(s)\ket{v_n(s), 0}=D(s)\ket{v_n(s)}\ket{0}=\ket{v_n(s), 0}$, we immediately have that $\ket{v_n(s), 0}$ is an eigenstate of $H(s)$ with eigenvalue $0$. 
For $k\neq n$, note that
\begin{align}
  V\ct(s) SV(s) \Pi_0 \ket{v_k(s), 0} &= \lambda_k(s) \ket{v_k(s), 0} + \sqrt{1-\lambda_k(s)^2} \ket{v_k(s), 0}^\perp, \\
  \Pi_0 V\ct(s) SV(s) \ket{v_k(s), 0} &= \lambda_k(s) \ket{v_k(s), 0}.
\end{align}
By combining these expressions we get
\begin{align}
  &H(s) \ket{v_k(s), 0}       =  i \sqrt{1-\lambda_k(s)^2} \ket{v_k(s), 0}^\perp, \\
  &H(s) \ket{v_k(s), 0}^\perp = -i \sqrt{1-\lambda_k(s)^2} \ket{v_k(s), 0},
\end{align}
where the second line follows from the fact that $H(s)$ is Hermitian and traceless. In other words, $H(s)$ acts on subspace $\mathcal{B}_k(s)$ as
$
  \sqrt{1-\lambda_k(s)^2} \sigma_y,
$
where $\sigma_y := \smx{0&-i\\i&0}$ is the Pauli $y$ matrix, which yields equation~(\ref{eq:Ek(s)}).

Since $\mathcal{B}^\perp(s)$ is the orthogonal complement of the union of invariant subspaces, it is also an invariant subspace for $H(s)$. Note that $H(s)$ restricted to this subspace is equal to zero, hence the remaining $n^2 - (2n - 1) = (n-1)^2$ eigenvalues of $H(s)$ are zero.
\end{proof}

\section{The quantum adiabatic theorem} \label{sect:Adiabatic theorem}

In adiabatic quantum computing it is a common practice to associate the intermediate state of the computation with the ground state (\ie, the lowest energy eigenstate) of the Hamiltonian. However, in our case the spectrum of $H(s)$ is symmetric about zero and the state that we are interested in lies in the middle of the spectrum. Hence, we will not use the ground state of $H(s)$, which has negative energy, but instead we will use the zero-eigenvector $\ket{\Psi_n(s)}$ given in equation (\ref{eq:En(s)}). Indeed, equation~(\ref{eq:v_n-pi}) shows that this state is closely related to the stationary distribution $\pi(s)$ of $P(s)$. In particular, the problem would be solved if we can reach the state $\ket{\Psi_n(1)}$, as measuring the first register of this state yields a vertex distributed according to $\pi(1)$, which only has support on marked vertices.

\subsection{The adiabatic condition}

We initially prepare the system in the zero-eigenvector $\ket{\Psi_n(0)}$ of $H(0)$ and then start to change the Hamiltonian $H(s)$ by slowly increasing the parameter $s$ from $0$ to $1$ according to some \emph{schedule} $s(t)$. If the schedule $s(t)$ is chosen so that it satisfies certain conditions, the system is guaranteed to stay close to the intermediate zero-eigenstate $\ket{\Psi_n(s)}$. 
Then, at $s=1$, the state will be close to $\ket{\Psi_n(1)}=\ket{v_n(1),0}$, where the first register only has overlap over marked vertices, so that a measurement yields a marked vertex with high probability.
Note that in our case the zero-eigenspace $\mathcal{B}_n(s) \cup \mathcal{B}^\perp(s)$ of the Hamiltonian $H(s)$ has a huge dimension, so we have to make sure that the non-relevant part $\mathcal{B}^\perp(s)$ is totally decoupled from $\ket{\Psi_n(s)}$ (the only zero-eigenvector that is relevant for our algorithm) before we apply the adiabatic condition. In particular, we want to show that
\begin{equation}
  \bra{\Phi_j(s)} \cdot \frac{d}{dt} \ket{\Psi_n(s)} = 0
  \label{eq:No leaking}
\end{equation}
for any $j \in \set{1, \dotsc, (n-1)^2}$, since this would imply that during the evolution $\ket{\Psi_n(s)}$ is not leaked into the subspace $\mathcal{B}^\perp(s)$ spanned by states $\ket{\Phi_j(s)}$. To see that this is indeed the case, note that
\begin{equation}
  \ket{\Phi_j(s)} \perp
  \spn \set{\ket{\Psi^\pm_1(s)}, \dotsc, \ket{\Psi^\pm_{n-1}(s)}, \ket{\Psi_n(s)}}
\end{equation}
for any $j \in \set{1, \dotsc, (n-1)^2}$, since the eigenvectors of $H(s)$ form an orthonormal basis. In particular,
\begin{align}
  \ket{\Phi_j(s)}
  &  \perp \spn
     \set{
       \ket{\Psi^{+}_1(s)} + \ket{\Psi^{-}_1(s)}, \dotsc,
       \ket{\Psi^{+}_{n-1}(s)} + \ket{\Psi^{-}_{n-1}(s)},
       \ket{\Psi_n(s)}
     } \nonumber \\
  &= \spn \set{\ket{v_1(s), 0}, \dotsc, \ket{v_{n-1}(s), 0}, \ket{v_n(s), 0}}.
\end{align}
Recall from equation (\ref{eq:En(s)}) that $\frac{d}{dt} \ket{\Psi_n(s)} = \frac{d}{dt} \ket{v_n(s)} \ket{0}$. Hence, the inner product in equation (\ref{eq:No leaking}) indeed vanishes. Thus, we can safely apply the adiabatic condition only for the relevant subspace in which the zero-eigenstate is not degenerate.

The \emph{folk} version of the Adiabatic Theorem~\cite{Mes} states that during the evolution the state of the system $\ket{\psi(t)}$ is guaranteed to stay close to the intermediate zero-eigenstate $\ket{\Psi_n(s)}$, more precisely,
\begin{equation}\label{eq:error}
  \forall t: \abs{\big\langle \Psi_n(s(t)) \big| \psi(t) \big\rangle}^2 \geq 1-\varepsilon^2,
\end{equation}
as long as the \emph{adiabatic condition}
\begin{equation}
  \forall t: \sum_{\sigma = \pm 1} \sum_{k = 1}^{n-1}
    \frac{\abs{\bra{\Psi_k^\sigma(s)} \cdot \frac{d}{dt} \ket{\Psi_n(s)}}^2}{\bigl( E^\sigma_k(s) - E_n(s) \bigr)^2}
    \leq \varepsilon^2
  \label{eq:Adiabatic condition}
\end{equation}
is satisfied.
While this condition is known not to be sufficient in full generality~(see \eg\ the discussion in \cite{jrs07}), we will assume that it can be applied in our setup. We will discuss about how this assumption may be suppressed in Section~\ref{sect:conclusion}.

If we insert all eigenvalues and eigenvectors from equations (\ref{eq:Ek(s)}) and (\ref{eq:En(s)}), the adiabatic condition (\ref{eq:Adiabatic condition}) can be written purely in terms of the eigenvalues and eigenvectors of the discriminant matrix $D(s)$:
\begin{equation}
  \forall t: \sum_\krange
    \frac{\abs{\bra{v_k(s)} \cdot \frac{d}{dt} \ket{v_n(s)}}^2}{1 - \lambda_k^2(s)}
    \leq \varepsilon^2.
  \label{eq:Adiabatic condition 2}
\end{equation}

\subsection{Rotation in a two-dimensional subspace} \label{sect:Rotation}

In this section we will provide a simple interpretation of the evolution of the eigenvector $\ket{v_n(s)}$. First, let us define the following superpositions over all elements, marked elements, and unmarked elements, respectively:
\begin{align}
 \ket{\pi}&:=\sum_{x\in X} \sqrt{\pi_x} \ket{x},&
 \ket{M}&:=\frac{1}{\sqrt{p_M}}\sum_{x\in M} \sqrt{\pi_x} \ket{x},&
 \ket{U}&:=\frac{1}{\sqrt{1-p_M}}\sum_{x\notin M} \sqrt{\pi_x} \ket{x}.
\end{align}
Now, we show that $\ket{v_n(s)}$ is subject to a rotation in the two-dimensional subspace $\Hr := \spn \set{\ket{U}, \ket{M}}$.
From equations~(\ref{eq:pi(s)}) and~(\ref{eq:v_n-pi}), we obtain
\begin{equation}
  \ket{v_n(s)}
  = \sum_{x\in X}\sqrt{\pi_x(s)}\ket{x}
  = \frac{1}{\sqrt{1-s(1-p_M)}} \Biggl(
      \sqrt{1-s}\sum_{x\notin M}\sqrt{\pi_x}\ket{x}+\sum_{x\in M}\sqrt{\pi_x}\ket{x}
    \Biggr).
\end{equation}
Using the definition of $\ket{U}$ and $\ket{M}$ we can rewrite this simply as
\begin{equation}
  \ket{v_n(s)} = \cos \theta(s) \ket{U} + \sin \theta(s) \ket{M},
  \label{eq:vn(s)}
\end{equation}
where
\begin{equation}
  \theta(s) := \arcsin \sqrt{\frac{p_M}{1-s(1-p_M)}}.
  \label{eq:theta(s)}
\end{equation}

Let us choose the schedule $s(t)$ so that the evolution of $\ket{v_n(s)}$ as defined by equation (\ref{eq:vn(s)}) corresponds to a rotation with constant angular velocity in the subspace $\Hr$. If $T$ is the length of the evolution and $s: [0,T] \to [0,1]$ is defined as
\begin{align}
  s(t) &:= \frac{1}{1-p_M} \left( 1-\frac{p_M}{\sin^2(\omega t + \theta_0)} \right), &
  \theta_0 &:= \arcsin \sqrt{p_M}, &
  \omega &:= \frac{1}{T} \arccos \sqrt{p_M},
  \label{eq:Parameters}
\end{align}
then
\begin{equation}
  \theta(s(t)) = \omega t + \theta_0.
  \label{eq:rotation}
\end{equation}

Let us choose a vector $\ket{v^\perp_n(s)}$ such that $\set{\ket{v(s)}, \ket{v^\perp_n(s)}}$ is an orthonormal basis of $\Hr$ for every $s$:
\begin{equation}
  \ket{v^\perp_n(s)} :=  
    -\sin \theta(s) \ket{U} +
    \cos \theta(s) \ket{M}.
  \label{eq:vnperp(s)}
\end{equation}
Then from equations (\ref{eq:vn(s)}) and  (\ref{eq:rotation}) we easily find that
\begin{equation}
  \frac{d}{dt} \ket{v_n(s(t))}
  = \omega \ket{v^\perp_n(s(t))}
  = \frac{1}{T} \arccos \sqrt{p_M} \ket{v^\perp_n(s(t))}.
  \label{eq:Derivative of vn(s)}
\end{equation}
Note that $\arccos \sqrt{p_M} \leq \frac{\pi}{2}$. Therefore, we can rewrite the adiabatic condition (\ref{eq:Adiabatic condition 2}) as follows:
\begin{equation}
  \forall s: \frac{\pi^2}{4\varepsilon^2} \sum_\krange
    \frac{\abs{\braket{v_k(s)}{v^\perp_n(s)}}^2}{1 - \lambda_k^2(s)}
    \leq T^2.
  \label{eq:Adiabatic condition 3}
\end{equation}
If this condition is satisfied, equation~(\ref{eq:error}) implies that we obtain at time $t=T$ a state $\ket{\psi(T)}$ close to $\ket{\Psi_n(1)}=\ket{v_n(1)}\ket{0}=\ket{M}\ket{0}$, so that measuring the first register yields a marked vertex with probability at least $1-\epsilon^2$.

\section{Running time of the quantum algorithm} \label{sect:Running time}

\subsection{Choice of running time $T$}

We have to change the parameter $s$ slowly for the evolution to be adiabatic, thus we want to choose $T$ big enough so that condition (\ref{eq:Adiabatic condition 3}) holds. Recall from Section~\ref{sect:D(s)} that we can assume that $\lambda_k(s) \geq 0$. Thus, $1 - \lambda^2_k(s) = \bigl(1 + \lambda_k(s)\bigr) \bigl(1 - \lambda_k(s)\bigr) \geq 1 - \lambda_k(s)$. Let us impose a slightly stronger condition on $T$ in equation (\ref{eq:Adiabatic condition 3}) by replacing $1 - \lambda^2_k(s)$ with $1 - \lambda_k(s)$. In addition, let us choose the smallest $T$ that still satisfies the inequality and use it as the running time of our adiabatic algorithm:\note{Can we\\interpret\\this directly\\without\\introducing\\$\HT(s)$ and\\$\MT(s)$?}
\begin{equation}
  T := \frac{\pi}{2\varepsilon} \max_{0 \leq s \leq 1}
    \sqrt{\sum_\krange \frac{\abs{\braket{v_k(s)}{v^\perp_n(s)}}^2}{1 - \lambda_k(s)}}.
  \label{eq:T}
\end{equation}

It turns out that there is an interesting relationship between this quantity and the hitting time of the Markov chain $P$.

\subsection{Hitting time of a Markov chain}


Let us first give an explicit expression for the hitting time $\HT(P,M)$ as defined in Section~\ref{sect:main-result}. Let $0_M$ and $1_M$ (resp., $0_U$ and $1_U$) be the all-zero and all-one row vectors of dimension $m$ (resp., $n-m$).
Furthermore, let $\tilde{\pi}_M:=\pi_M/p_M$ and $\tilde{\pi}_U:=\pi_U/(1-p_M)$ be the row vectors describing distributions over marked and unmarked vertices.
Then, the distribution of vertices at the the first execution of step~\ref{item:walk} of $\RandomWalkAlgorithm$ is $\rv{\tilde{\pi}_U}{0_M}$, and from the definition of the discriminant $D(1)$, we have
\begin{eqnarray}
  \HT(P,M)
  &:=& \sum_{t=0}^\infty \Pr[\textrm{No marked vertex found after $t$ steps from $\rv{\tilde{\pi}_U}{0_M}$}] \\
  & =& \sum_{t=0}^\infty \rv{\tilde{\pi}_U}{0_M} \, P^t(1) \, \rv{1_U}{0_M}\tp \\
  & =& \sum_{t=0}^\infty \bra{U} D^t(1) \ket{U},
\end{eqnarray}
where the last equality follows from equation (\ref{eq:D(1)}). We will show that the running time of our adiabatic quantum algorithm is directly related to the square root of the hitting time $\HT(P,M)$. In order to do this, we define the following extension of the hitting time. Let
\begin{equation}
  \HT(s) := \sum_{t=0}^\infty \bra{U} \bigl( D^t(s)-\proj{v_n(s)} \bigr) \ket{U}.
\end{equation}
Note that $\HT(1) = \HT(P,M)$ since $\braket{U}{v_n(1)} = 0$. This justifies to consider $\HT(s)$ as an extension of the hitting time. Intuitively, $\HT(s)$ may be understood as the time it takes for $P(s)$ to converge to its stationary distribution, starting from $\rv{\tilde{\pi}_U}{0_M}$. For $s=1$, the walk $P(1)=P'$ converges to the (non-unique) stationary distribution $\rv{0_U}{\tilde{\pi}_M}$, which only has support over marked elements.

Using the expansion $(1-x)^{-1} = \sum_{t=0}^\infty x^t$ and the spectral decomposition~(\ref{eq:spectral}) of the discriminant $D(s)$, we have
\begin{eqnarray}
  \HT(s)
  &=& \sum_{k \neq n} \sum_{t=0}^\infty \lambda_k^t(s) \braket{U}{v_k(s)} \braket{v_k(s)}{U} \\
  &=& \sum_\krange \frac{\abs{\braket{v_k(s)}{U}}^2}{1 - \lambda_k(s)}
  \label{eq:HT(s)}.
\end{eqnarray}

\subsection{Relation between the running time and the extended hitting times}

Let us express the running time $T$ from equation (\ref{eq:T}) in terms of $\HT(s)$. Define
\begin{equation}
  A(s) := \sum_\krange \frac{\ket{v_k(s)}\bra{v_k(s)}}{1 - \lambda_k(s)}.
  \label{eq:A}
\end{equation}
Note that both $T$ and $\HT(s)$ can be expressed in terms of $A(s)$ as follows:
\begin{align}
  T &= \frac{\pi}{2\varepsilon} \max_{0 \leq s \leq 1} \sqrt{\bra{v^\perp_n(s)} A(s) \ket{v^\perp_n(s)}}, &
  \HT(s) &= \bra{U} A(s) \ket{U}.
  \label{eq:T and HT(s)}
\end{align}
By definition, we have $A(s)\ket{v_n(s)}=0$, which together with equation~(\ref{eq:vn(s)}) implies that
\begin{equation}
  A(s)\ket{M}=-\frac{\cos\theta(s)}{\sin\theta(s)}A(s)\ket{U}.
  \label{eq:AM-AU}
\end{equation}
Using this and the definition of $\ket{v^\perp_n(s)}$ in equation~(\ref{eq:vnperp(s)}), we see that $\bra{v^\perp_n(s)} A(s) \ket{v^\perp_n(s)} = \bra{U} A(s) \ket{U} / \sin^2\theta(s)$. Thus we get the following relationship between $T$ and $\HT(s)$:
\begin{equation}
  T = \frac{\pi}{2\varepsilon} \max_{0 \leq s \leq 1} \frac{\sqrt{\HT(s)}}{\sin\theta(s)}.
  \label{eq:T in terms of HT(s)}
\end{equation}

To relate $T$ and the usual hitting time $\HT(P,M)$ of $P$, we first provide an explicit expression for $\HT(s)$ in terms of $\HT(P,M)$ (the proof is given in the appendix).
\begin{lemma}\label{lem:HT}
$\HT(s) = \HT(P,M) \, \sin^4 \theta(s)$.
\end{lemma}
Now, recalling the definition of $\theta(s)$ in equation (\ref{eq:theta(s)}), it is easy to see that the maximum in equation (\ref{eq:T in terms of HT(s)}) is attained at $s = 1$. This immediately implies that the running time of the adiabatic quantum algorithm is given by
\begin{equation}
  T = \frac{\pi}{2\varepsilon} \sqrt{\HT(P,M)},
\end{equation}
therefore providing a quadratic improvement over the classical hitting time. This also concludes the proof of Theorem~\ref{thm:main}.

\section{Conclusion and discussion\label{sect:conclusion}}

Our quantum algorithm defines a new notion of quantum hitting time, which is quadratically smaller than the classical hitting time for any reversible Markov chain and any set of marked elements. While previous approaches were subject to various restrictions, \eg, the quantum algorithm could only detect the presence of marked elements~\cite{Sze}, did not always provide full quadratic speed-up~\cite{MNRS}, or could only be applied for state-transitive Markov chains with a unique marked element~\cite{MNRS2}, our adiabatic approach only requires minimal assumptions. Indeed, it can be shown that the only remaining condition, reversibility, is necessary. Let us consider the Markov chain on a cycle $P=\frac{I+C}{2}$, where $C$ implements a clockwise shift, \ie, $C\ket{x}=\ket{(x+1) \mod n}$. This Markov chain is ergodic but not reversible. While its classical hitting time is of order $\Theta(n)$, a simple locality argument implies that any quantum operator acting locally on the cycle requires a time $\Omega(n)$ to find a marked vertex, so that a quadratic speed-up cannot be achieved. Magniez \etal\ ~\cite{MNRS2} have also shown that under reasonable conditions the quadratic speed-up is optimal. This provides evidence that our result is both as strong and as general as possible.

While our result relies on the assumption that the \emph{folk} adiabatic condition is sufficient, this assumption could be suppressed in different ways. One option would be to actually prove that the \emph{folk} Adiabatic Theorem holds in our setup, as was previously done for the adiabatic version of Grover's algorithm~\cite{jrs07}. Another option would be to circumvent adiabatic evolution altogether, by using the phase randomization technique of Boixo~\etal ~\cite{bks09}. Their technique provides a quantum circuit realizing the same evolution as the adiabatic approach with a similar running time, but without relying on the adiabatic condition. This leads to the quantum circuit algorithm described in~\cite{kmor10}.

Finally, note that in order to design the schedule $s(t)$, our algorithm requires to know $p_M$ and the order of magnitude of $\HT(P,M)$. These assumptions are standard in other quantum algorithms for this problem. In particular, a similar issue arises in Grover's algorithm when the number of marked elements is unknown. In Grover's case, there are techniques to deal with this issue~\cite{BoyerBHT98}, and similar techniques could be applied in our case. While we do not provide a full answer to these questions in the present paper, they do not present any new technical difficulty and we refer the reader to~\cite{kmor10} where a full study of these implementation issues is provided for a related quantum circuit algorithm.

\section*{Acknowledgments}
J. Roland would like to thank F. Magniez, A. Nayak, M. Santha and R. Somma for useful discussions. H. Krovi would like to thank F. Magniez for useful discussions.
This research has been supported in part by ARO/NSA under grant W911NF-09-1-0569.
M. Ozols acknowledges support from QuantumWorks.


\bibliography{adiabatic-quantum-search}

\appendix*
\section{Proof of Lemma~\ref{lem:HT}\label{app:HT}}

In this section, we will often use $\dot{f}(s)$ as a shorthand form of $\frac{d}{ds} f(s)$. We will show that the derivative of $\HT(s)$ satisfies the following lemma.

\begin{lemma}\label{lemma:derivative-HT}
The derivative of $\HT(s)$ is related to $\HT(s)$ as:
\begin{equation}
 \frac{d}{ds} \HT(s)=4\dot{\theta}(s)\frac{\cos\theta(s)}{\sin\theta(s)}\HT(s).
\end{equation}
\end{lemma}

Note that Lemma~\ref{lem:HT} follows directly from there, since $\HT(s) = \HT(P,M) \, \sin^4\theta(s)$ satisfies the differential equation for $\HT(s)$ with the boundary condition $\HT(1)=\HT(P,M)$. Therefore, it remains to prove Lemma~\ref{lemma:derivative-HT}.

Before proving Lemma~\ref{lemma:derivative-HT}, let us consider the derivative of the discriminant $D(s)$. Let $\Pi_M=\sum_{x\in M}\proj{x}$ be the projector onto the $m$-dimensional subspace spanned by marked vertices, and let $\{X,Y\}:=XY+YX$ be the anticommutator of $X$ and $Y$.

\begin{lemma}\label{lem:d-dot}
$  \dot{D}(s) = \frac{1}{2(1-s)} \bigl\{ \Pi_M, I-D(s) \bigr\}$.
\end{lemma}
\begin{proof}
Note from equation (\ref{eq:Block P(s)}) that $D(s)$ has the following block structure:
\begin{equation}
  D(s) = \mx{
    \sqrt{P_{UU} \circ P\tp_{UU}} & \sqrt{(1-s) (P_{UM} \circ P\tp_{MU})} \\
    \sqrt{(1-s) (P_{MU} \circ P\tp_{UM})} & (1-s) \sqrt{P_{MM}\circ P\tp_{MM}} + s I }.
  \label{eq:Block D(s)}
\end{equation}
Hence,
\begin{equation}
  \dot{D}(s) = \mx{
    0 & -\frac{1}{2\sqrt{1-s}} \sqrt{P_{UM} \circ P\tp_{MU}} \\
    -\frac{1}{2\sqrt{1-s}} \sqrt{P_{MU} \circ P\tp_{UM}} & I -  \sqrt{P_{MM}\circ P\tp_{MM}}}.
\end{equation}
Observe that $\dot{D}(s)+\frac{1}{2(1-s)}\{\Pi_M,D(s)\}=\frac{1}{1-s}\Pi_M$, which implies the lemma.
\end{proof}

We are now ready to prove Lemma~\ref{lemma:derivative-HT}.
\begin{proof}[Proof of Lemma~\ref{lemma:derivative-HT}]
In this proof we will often omit to write the dependence on $s$ explicitly. From equation (\ref{eq:T and HT(s)}) we have
\begin{equation}
  \frac{d}{ds} \HT\ofs = \bra{U} \dot{A}\ofs \ket{U}.
  \label{eq:dHT}
\end{equation}
Note that $A\ofs = B\ofs^{-1}-\Pi_n\ofs$, where
\begin{align}
  B\ofs &:= I - D\ofs + \Pi_n\ofs, &
  \Pi_n\ofs := \proj{v_n\ofs}.
\end{align}
For any invertible matrix $M\ofs$ we have $\frac{d}{ds} (M^{-1}) = -M^{-1} \dot{M} M^{-1}$. Therefore,
\begin{equation}
  \dot{A}\ofs = - B\ofs^{-1} \bigl( -\dot{D}\ofs+\dot{\Pi}_n\ofs \bigr) B\ofs^{-1} - \dot{\Pi}_n\ofs.
\end{equation}
Hence, we have $\frac{d}{ds} \HT\ofs = h_1 + h_2 + h_3$, where
\begin{align}
  h_1 &:= \bra{U} B\ofs^{-1} \dot{D}\ofs B\ofs^{-1} \ket{U}, \\
  h_2 &:=-\bra{U} B\ofs^{-1} \dot{\Pi}_n\ofs B\ofs^{-1} \ket{U}, \\
  h_3 &:=-\bra{U} \dot{\Pi}_n\ofs \ket{U}.
\end{align}
Let us evaluate each of these terms separately. We can use Lemma~\ref{lem:d-dot} and the definition of $B$ to express the first term $h_1$ as follows:
\begin{align}
  2(1-s) h_1
   &= \bra{U} B\ofs^{-1} \{\Pi_M,I-D\ofs\} B\ofs^{-1} \ket{U} \\
   &= \bra{U} B\ofs^{-1} \bigl(\{\Pi_M,B\ofs\}-\{\Pi_M,\Pi_n\ofs\}\bigr) B\ofs^{-1} \ket{U} \\
   &= \bra{U} \{B\ofs^{-1},\Pi_M\} \ket{U} - \bra{U} B\ofs^{-1} \{\Pi_M,\Pi_n\ofs\} B\ofs^{-1} \ket{U}.
\end{align}
Note that $\Pi_M \ket{U} = 0$, thus the first term vanishes. Also note that $B\ofs^{-1} \ket{v_n\ofs} = \ket{v_n\ofs}$ and $B\ofs^{-1}$ is Hermitian, thus
\begin{equation}
  2(1-s) h_1 = -2 \bra{U} B\ofs^{-1} \Pi_M \Pi_n\ofs \ket{U}.
  \label{eq:h1.2}
\end{equation}
Note from equation (\ref{eq:vn(s)}) that $\Pi_M \ket{v_n\ofs} = \sin \theta\ofs \ket{M}$. Moreover, using equation (\ref{eq:AM-AU}) we can simplify equation (\ref{eq:h1.2}) even further:
\begin{align}
 2(1-s) h_1
  &= -2 \cos \theta\ofs \sin \theta\ofs \bra{U} B\ofs^{-1} \ket{M} \\
  &= -2 \cos \theta\ofs \sin \theta\ofs \bra{U} (A\ofs + \Pi_n\ofs) \ket{M} \\
  &= -2 \cos \theta\ofs \sin \theta\ofs \bra{U} A\ofs \ket{M} - 2 \cos^2\theta\ofs \sin^2\theta\ofs \\
  &=  2 \cos^2 \theta\ofs \bigl( \bra{U} A\ofs \ket{U} - \sin^2\theta\ofs \bigr).
\end{align}

Let us now consider the second term $h_2$:
\begin{align}
 h_2
  &= -\bra{U} B\ofs^{-1} \dot{\Pi}_n\ofs B\ofs^{-1} \ket{U} \\
  &= -\cos \theta\ofs \bigl( \bra{\dot{v}_n\ofs} B\ofs^{-1} \ket{U} + \bra{U} B\ofs^{-1} \ket{\dot{v}_n\ofs} \bigr) \\
  &= -2 \cos \theta\ofs \bra{U} \bigl( A\ofs + \Pi_n\ofs \bigr) \ket{\dot{v}_n\ofs} \\
  &= -2 \dot{\theta}\ofs \cos \theta\ofs \bigl(-\sin \theta\ofs \bra{U} A\ofs \ket{U} + \cos \theta\ofs \bra{U} A\ofs \ket{M}\bigr) \\
  &= -2 \dot{\theta}\ofs \cos \theta\ofs \left(-\sin \theta\ofs \bra{U} A\ofs \ket{U} - \frac{\cos^2\theta\ofs}{\sin\theta\ofs} \bra{U} A\ofs \ket{U} \right) \\
  &=  2 \dot{\theta}\ofs \frac{\cos \theta\ofs}{\sin \theta\ofs} \bra{U} A\ofs \ket{U},
\end{align}
where we have used the fact that
\begin{equation}
  \ket{\dot{v}_n\ofs} = \dot{\theta}\ofs \bigl( -\sin \theta\ofs \ket{U} + \cos \theta\ofs \ket{M} \bigr),
\end{equation}
and $\Pi_n \ket{\dot{v}_n\ofs} = 0$. Similarly, for the last term $h_3$ we have
\begin{equation}
  h_3 = -\bra{U} \dot{\Pi}_n\ofs \ket{U}
      = 2 \dot{\theta}\ofs \cos \theta\ofs \sin \theta\ofs.
\end{equation}

Putting all the terms back together, we have
\begin{align}
  \label{eq:UAU}
  \bra{U}\dot{A}\ket{U}
   &= \frac{\cos^2\theta}{1-s} \bigl(\bra{U}A\ket{U}-\sin^2\theta\bigr)
      +2\dot{\theta}\frac{\cos\theta}{\sin\theta}\bra{U}A\ket{U}+2\dot{\theta}\cos\theta\sin\theta \\
   &= \frac{\cos\theta}{\sin\theta}\left(2\dot{\theta}+\frac{\cos\theta\sin\theta}{1-s}\right)\bra{U}A\ket{U}
      +\cos\theta\sin\theta\left(2\dot{\theta}-\frac{\cos\theta\sin\theta}{1-s}\right). \nonumber
\end{align}
From the definition of $\theta$ in equation (\ref{eq:theta(s)}) it is straightforward to check that
\begin{equation}
  2\dot{\theta}
  = \frac{\cos\theta\ofs\sin\theta\ofs}{1-s}
\end{equation}
which together with equations (\ref{eq:dHT}) and (\ref{eq:UAU}) implies the lemma.
\end{proof}

\end{document}